\documentclass[a4paper,UKenglish,cleveref, autoref, thm-restate]{lipics-v2021}
\usepackage{braket}

\title{Quantum Advantage with Shallow Circuits Under Arbitrary Corruption} 

\titlerunning{Quantum Advantage with Shallow Circuits Under Arbitrary Corruption} 

\author{Atsuya Hasegawa}{The University of Tokyo, Japan}
{atsuyahasegawa@is.s.u-tokyo.ac.jp}{}{}

\author{Fran{\c{c}}ois Le Gall}{Nagoya University, Japan}{legall@math.nagoya-u.ac.jp}{}{}

\authorrunning{A.\,Hasegawa and F.\,Le Gall} 

\Copyright{Atsuya Hasegawa and Fran{\c{c}}ois Le Gall} 



\category{} 

\relatedversion{} 

\funding{JSPS KAKENHI grants Nos.~JP16H01705, JP19H04066, JP20H00579, JP20H04139 and by the MEXT Quantum Leap Flagship Program (MEXT Q-LEAP) grant No.~JPMXS0120319794.}

\acknowledgements{AH is grateful to Hidefumi Hiraishi and Hiroshi Imai for helpful discussions and continuous supports. FLG would also like to thank Robert K{\"o}nig for useful discussions.}

\nolinenumbers 

\hideLIPIcs  

\EventEditors{Hee-Kap Ahn and Kunihiko Sadakane}
\EventNoEds{2}
\EventLongTitle{32nd International Symposium on Algorithms and Computation (ISAAC 2021)}
\EventShortTitle{ISAAC 2021}
\EventAcronym{ISAAC}
\EventYear{2021}
\EventDate{December 6--8, 2021}
\EventLocation{Fukuoka, Japan}
\EventLogo{}
\SeriesVolume{212}
\ArticleNo{31}

\begin{document}

\maketitle

\begin{abstract}
Recent works by Bravyi, Gosset and K{\"o}nig (Science 2018), Bene Watts et al. (STOC 2019), Coudron, Stark and Vidick (QIP 2019) and Le Gall (CCC 2019) have shown unconditional separations between the computational powers of shallow (i.e., small-depth) quantum and classical circuits: quantum circuits can solve in constant depth computational problems that require logarithmic depth to solve with classical circuits. Using quantum error correction, Bravyi, Gosset, K{\"o}nig and Tomamichel (Nature Physics 2020) further proved that a similar separation still persists even if quantum circuits are subject to local stochastic noise.

In this paper, we consider the case where \emph{any} constant fraction of the qubits (for instance, huge blocks of qubits) may be arbitrarily corrupted at the end of the computation.
We make a first step forward towards establishing a quantum advantage even in this extremely challenging setting: we show that there exists a computational problem that can be solved in constant depth by a quantum circuit but such that even solving \emph{any} large subproblem of this problem requires logarithmic depth with bounded fan-in classical circuits. This gives another compelling evidence of the computational power of quantum shallow circuits. 

In order to show our result, we consider the Graph State Sampling problem (which was also used in prior works) on expander graphs. We exploit the ``robustness'' of expander graphs against vertex corruption to show that a subproblem hard for small-depth classical circuits can still be extracted from the output of the corrupted quantum circuit. 
\end{abstract}

\section{Introduction}
\noindent{\textbf{Background.}}
Quantum computing was introduced in the early 1980s as a quantum mechanical model of the Turing machine that has a potential to simulate things that a classical computer could not \cite{benioff1980computer,feynman1982simulating}. In 1994, Peter Shor developed a polynomial-time quantum algorithm for factoring integers \cite{Shor_1997}, which gave an exponential speedup over the most efficient known classical algorithms for this task.

While initially realizing a physical quantum computer was thought to be extremely challenging, nowadays various types of high-fidelity processors capable of quantum algorithms have been developed \cite{benhelm2008towards,wallraff2004strong,nakamura1997spectroscopy,o2007optical,you2011atomic}. These devices with noie and relatively small scale are called NISQ (Noisy Intermediate-Scale Quantum) devices \cite{preskill2018quantum}. For such devices, ``quantum supremacy'' \cite{preskill2013quantum} has been recently reported \cite{arute2019quantum,Zhongeabe8770}.

While Shor's algorithm and such quantum supremacy results strongly suggest that  quantum computation is more powerful than classical computation, they are not mathematical proofs. While the superiority of quantum computation has been formally shown in constrained models such as query complexity \cite{ams18} and communication complexity \cite{de2002quantum,lo2000classical} and considering complexity classes relative to an oracle \cite{bernstein1997quantum,raz2019oracle}, almost no definite answer is known in standard computational models such as Turing machines or general circuits. Since the complexity class BQP (the class of the problems that can be solved efficiently by a quantum computer) satisfies the inclusions P $\subseteq$ BPP $\subseteq$ BQP $\subseteq$ PSPACE, an unconditional separation between BPP and BQP would imply a separation between P and PSPACE, which would be a significant breakthrough. Therefore,  unconditional separations between the computational powers of quantum computers and classical computers in a general setting are expected to be very hard to obtain.

However, with several assumptions from computational complexity such as non-collapse of the polynomial hierarchy or some conjectures on the hardness of the permanent, the superiority of quantum computation has been shown in the circuit model: even approximate or noisy probabilistic distributions of small depth quantum circuits are hard to simulate for classical computers \cite{aaronson2011computational,10.5555/2685179.2685186,10.5555/3135595.3135617,bouland_et_al:LIPIcs:2018:10108,bremner2011classical,bremner2016average,bremner2017achieving,terhal2004adptive}. A recent breakthrough by Bravyi, Gosset and K{\"o}nig \cite{Bravyi_2018} showed an \emph{unconditional} separation between the computational powers of small-depth quantum and classical circuits: they constructed a computational problem that can be solved by quantum circuits of constant depth composed of one- and two-qubit gates acting locally on a grid and showed that any probabilistic classical circuit with bounded fan-in gates solving this problem on all inputs must have depth $\Omega(\log n)$, where $n$ denotes the input size. The computational problem they use is a relation problem (i.e., for any input there are several possible outputs). Besides its theoretical importance, this separation is also important since shallow quantum circuits are likely to be easy to implement on physical devices experimentally due to their robustness to noise and decoherence.

There are several results related to this separation. Coudron, Stark and Vidick \cite{coudron2021trading} and Le Gall \cite{le2019average} showed a similar separation in the average case setting, instead of the worst case setting considered in the original version of \cite{Bravyi_2018}: there exists a relation problem such that constant-depth quantum circuits can solve the relation on all inputs, but any $O(\log n)$ depth randomized bounded fan-in classical circuits cannot solve it on {\sl most} inputs with high probability. Bene Watts et al.~\cite{watts2019exponential} showed that a similar separation holds against classical circuits using unbounded fan-in gates and, considering interactive tasks, Grier and Schaeffer~\cite{grier2020interactive} showed even stronger classical lower bounds.

Bravyi et al.~\cite{bravyi2020quantum} additionally proved, using quantum error-correction, that a similar separation holds even if quantum circuits are corrupted by local stochastic noise (see Definition~\ref{def:noise} below). The computational problem used in \cite{bravyi2020quantum} is a generalized version, defined on a 3D grid, of the magic square game \cite{mermin1990extreme,peres1990incompatible}, which is a nonlocal game with two cooperating players Alice and Bob who cannot communicate. The noise model is as follows.
\begin{definition}[Definition 9 in \cite{bravyi2020quantum}]\label{def:noise}
  Consider a random $n$-qubit Pauli error $E \in \{I,X,Y,Z\}^{\otimes n}$ and let Supp($E$) $\subseteq$ $[n]$ denote its support, i.e., the subset of qubits acted on by either $X,Y,$ or $Z$. For any constant p $\in$ [0,1], E is called p-local stochastic noise if 
  \[
    \mathrm{Pr}[F \subseteq Supp(E)] \leq p^{|F|} \hspace{10pt} \ \mathrm{for \ all} \ F \subseteq [n].
  \]
\end{definition}
In this model, an error is given as applying a gate which is $X,Y$ or $Z$. This definition assumes that when picking an arbitrary subset of qubits, the probability of all qubits are corrupted is an exponentially small function of the size of the subset. This property, which implies that a subset of size $\Omega(\log n)$ contains with probability $1-1/\textrm{poly}(n)$ at least one qubit that is not corrupted, is crucial in \cite{bravyi2020quantum} to use quantum error correction. Note that, considering interactive tasks, Grier, Ju and Schaeffer \cite{grier2021interactive} showed even stronger classical lower bounds in the same noise setting.\vspace{2mm}

\noindent{\textbf{The Graph State Sampling problem.}}
Before presenting our results, let us describe in more details the computational problem introduced in \cite{Bravyi_2018}, which is called the 2D Hidden Linear Function problem and corresponds to an extension of the Bernstein-Vazirani problem~\cite{bernstein1997quantum}.

We actually describe a slightly more general computational problem that we name the ``Graph State Sampling problem'' and denote $\rho(G)$. Here $G=(V,E)$ is a graph specifying the problem. For any graph~$G$, the problem $\rho(G)$ is a relation $\rho(G)\subseteq \{0,1\}^{|V|+|E|} \times \{0,1\}^{|V|}$. Given an input $x\in \{0,1\}^{|V|+|E|}$ for this relation problem, which we interpret as a pair $x=(y,H)$ with $y\in\{0,1\}^{|V|}$ and $H$ being a subgraph of~$G$, we ask to output any bit string $z\in\{0,1\}^{|V|}$ that may appear with nonzero probability when measuring the graph state corresponding to the subgraph $H$ in a basis determined by the bit string $y$. We refer to Section~\ref{def:GSS} for details of the definition of the problem.

The 2D Hidden Linear Function problem considered in \cite{Bravyi_2018} is essentially the problem $\rho(G)$ where $
G$ is the family of 2D grid graphs. Since the graph states of subgraphs of a 2D grid can be constructed by constant-depth quantum circuits whose gates act locally on the grid graphs, the problem $\rho(G)$ can be solved by a constant-depth quantum circuit. At the same time, Bravyi et al.~\cite{Bravyi_2018} prove that no small-depth classical circuits can solve this problem using an argument based on the existence of quantum nonlocality in a triangle (first shown by Barrett et al.~\cite{barrett2007modeling}).
The main result in \cite{Bravyi_2018} can essentially be restated as follows.\footnote{In this paper, for any graph $G$ we use the notation $|G|$ to denote the size of the vertex set of $G$.}

\begin{theorem}[\cite{Bravyi_2018}]\label{th:B18}
  There exist a constant $\alpha > 0$ such that the following holds for all sufficiently large 2D grid graphs $G$:
  \begin{enumerate}[(i)]
      \item $\rho(G)$ can be solved on all inputs with certainty by a constant-depth quantum circuit on $\Theta(|G|)$ qubits composed of one- and two-qubit gates;
      \item no bounded-fanin classical probabilistic circuit whose depth is less than $\alpha \log(|G|)$ can solve with high probability $\rho(G)$ on all inputs.
  \end{enumerate}
\end{theorem}

\noindent{\textbf{Description of our results.}}
In this paper we show the following result.

\begin{theorem}\label{main}
      There exist constants $\alpha > 0$ and $\epsilon > 0$, and a family of graphs $(G_i)_{i \in \mathbb{N}}$ with $\displaystyle \lim_{i \to \infty} |G_i| = \infty$ such that the following holds for all sufficiently large $i$:
      \begin{enumerate}[(i)]
          \item $\rho(G_i)$ can be solved on all inputs with certainty by a constant-depth quantum circuit on $\Theta(|G_i|)$ qubits composed of one- and two-qubit gates;
          \item for any induced subgraph $S_i$ of $G_i$ such that $|S_i| \geq (1-\epsilon) |G_i|$, no bounded-fanin classical probabilistic circuit whose depth is less than $\alpha \log(|S_i|)$ can solve with high probability $\rho(S_i)$ on all inputs.
      \end{enumerate}
\end{theorem}
Item (ii) of Theorem \ref{main}, which is proved by considering the Graph State Sampling problem over expander graphs, gives a significantly stronger hardness guarantee than in Theorem \ref{th:B18} and thus provides us further compelling evidence of the computational power of quantum shallow circuits. 

We stress that Theorem \ref{main} does not claim a quantum advantage for noisy shallow quantum circuits: Theorem \ref{main} simply shows that there exists a problem that can be computed by shallow quantum circuit but such that a shallow classical circuit cannot solve \emph{any} (large) subproblem of it. We can nevertheless interpret Item (ii) as follows. A quantum circuit~$\mathcal{C}$ solving the relation $\rho(G_i)$ has~$|G_i|$ output qubits, which are measured at the end of the computation to give the output string $z\in\{0,1\}^{|G_i|}$ that is a solution for the relation. 
Assume that an adversary chooses up to $\epsilon|G_i|$ qubits among these $|G_i|$ qubits and corrupts them in an arbitrary way (or, essentially equivalently, corrupts the bits of the measurement outcomes corresponding to these positions). Let $S_i$ denote the set of qubits that are not corrupted by the adversary. Since $\rho(S_i)$ corresponds to a subproblem\footnote{This property can immediately be derived from the formal definition of the computational problem given in Section \ref{def:GSS}} of $\rho(G_i)$, and since $\mathcal{C}$ before the corruption solved the problem $\rho(G_i)$ on all inputs, even after the corruption the part of the output of $\mathcal{C}$ corresponding to the qubits in $S_i$ gives a correct solution to the problem $\rho(S_i)$. On the other hand, Item (ii) of Theorem \ref{main} shows that no small-depth classical circuit can solve $\rho(S_i)$. (Note that in Item (ii) we even allow the classical circuit to depend on $S_i$.)

Let us compare this model of corruption of qubits with the model of noise considered in~\cite{bravyi2020quantum}. As already mentioned, in the error model of \cite{bravyi2020quantum} (Definition 1 above), the probability that all qubits in a given set of size $\Theta(\log n )$, where $n$ denotes the total number of qubits, are corrupted by the noise is polynomially small and thus can be neglected. In comparison, Theorem \ref{main} deals with the case where any subset of qubits of size as large as $\Theta(n)$ can be corrupted, and shows that the quantum advantage is still preserved in this case. In this sense, our result gives a further compelling evidence of the computational power of quantum shallow circuits. \vspace{2mm}

\noindent{\textbf{Brief overview of our techniques and organization of the paper.}}
The family of graph $(G_i)_{i \in \mathbb{N}}$ used to prove Theorem \ref{main} is a class of expander graphs of constant degree.\footnote{For technical reasons, we actually consider a family of graphs of the form $G\times K_2$, where $G$ is an expander graph of constant degree and $K_2$ is the graph consisting of a single edge.} Item (i) of Theorem \ref{main} essentially follows from the fact the graph state of a bounded-degree graph can be created by a constant-depth quantum circuit composed of one- and two-qubit gates. The proof of Item (ii) of Theorem \ref{main} exploits the ``robustness'' of expander graphs against corruption of vertices. More precisely, we show that even after corrupting a constant fraction of vertices, an expander graph still has a large grid minor (see Lemma \ref{minorinexpander} in Section~3). Finally, exploiting the existence of a large grid minor, we can use arguments based on quantum nonlocality (very similarly to the arguments used in prior works \cite{Bravyi_2018,le2019average}) to conclude that any classical circuit solving the Graph State Sampling problem on the expander graph requires logarithmic depth. 

After giving preliminaries in Section 2, in Section 3 we present graph-theoretical results about expander graphs. In particular, we prove Lemma \ref{minorinexpander} about the existence of large grid minors in corrupted expander graphs. In Section 4, we review the result about the nonlocality of a triangle quantum graph state used in prior works. In Section 5, we define our computational problem and prove Theorem \ref{main}. 
\section{Preliminaries}
\noindent{\bf Graph theory.}
All graphs considered in this paper are undirected. We write a graph as $G = (V,E)$ where $V$ is the vertex set and $E$ is the edge set. $|G|$ means the number of vertex of graph G, i.e., $|V|$. The degree of a vertex is the number of edges that are incident to the vertex. We denote $deg(v)$ the degree of a vertex $v$. Given a graph $G = (V,E)$ and any vertex set $U \subseteq V$, we denote $N_G(U) = \{v \in V \setminus U : v \mathrm{\ has \ a \ neighbor \ in \ }U\}$ the external neighborhood of $U$ in $G$. 

Let us describe the definition of graph minors. Graph $\Gamma$ is a graph minor of graph $G$ if it is isomorphic to a graph obtained from $G$ by deleting edges and vertices and by contracting edges. Note that if a graph $\Gamma$ is a minor of a subgraph $S$ of a graph $G$, $\Gamma$ is also a minor of $G$. The definition is equivalent to the following definition, that if graph $\Gamma$ is a minor of graph $G$, we can decompose $G$ to connected subgraphs which connect to each other like $\Gamma$. We will use this definition later for explicit explanations.

\begin{definition}[Minor, Definition 1 in \cite{krivelevich2009minors}]
  A graph $\Gamma$ is a minor of a graph G if for every vertex u $\in$ $\Gamma$ there is a connected subgraph $G_u$ of G such that all subgraphs $G_u$ are vertex disjoint, and G contains an edge between $G_u$ and $G_{u'}$ whenever $\{u,u'\}$ is an edge of $\Gamma$.
\end{definition}

Next, we will define the product of graphs $G \times H$ of graphs $G = (V_G, E_G)$ and $H = (V_H, E_H)$. The vertex set of $G \times H$ is the Cartesian product $V_G \times V_H$ and an edge is spanned between $(u_G,u_H)$ and $(v_G,v_H)$ if and only if $u_G = v_G$ and $\{u_H, v_H\} \in E_H$, or $u_H = v_H$ and $\{u_G, v_G\} \in E_G$. There are several ways to define graph products but we will use the definition above. In this paper, we particularly use $G \times K_2$, which $K_2$ is the complete graph of two vertices. Given a graph $G = (V,E)$, the graph product $G \times K_2$ is with $2|V|$ vertices and $2|E|+|V|$ edges.

Lastly, we refer to the Vizing's theorem, which is about edge coloring. Edge coloring is to assign colors to edges so that the edges of the same color are not incident. 
\begin{lemma}[Vizing's theorem \cite{diestel2000graduate}]\label{edgecoloring}
  Every simple undirected graph can be edge colored using a number of colors that is the maximum degree or the maximum degree+1. 
\end{lemma}

\noindent{\bf Quantum circuits.}
The textbook \cite{nielsen_chuang_2000} is a good reference about notations of quantum computation in our paper.
We will use the Pauli $X$, $Y$ and $Z$ gates, the Hadamard gate and the $S$ and $T$ gates as single qubit gates:
\[\scriptsize
  X = 
  \begin{pmatrix}
      0 & 1 \\
      1 & 0 \\
    \end{pmatrix}\!,\ 
  Y = 
    \begin{pmatrix}
      0 & -i \\
      i & 0 \\
    \end{pmatrix}\!,\ 
  Z = 
    \begin{pmatrix}
      1 & 0 \\
      0 & -1 \\
    \end{pmatrix}\!,\
  H = \frac{1}{\sqrt{2}}
    \begin{pmatrix}
      1 & 1 \\
      1 & -1 \\
    \end{pmatrix}\!,\
  S =
    \begin{pmatrix}
      1 & 0 \\
      0 & -i \\
    \end{pmatrix}\!,\
  T=
    \begin{pmatrix}
      1 & 0 \\
      0 & e^{i\pi/4} \\
    \end{pmatrix}\!,
\]
where $i$ denotes the imaginary unit of complex numbers. (Note the Phase gate $S$ differs from the standard one in \cite{nielsen_chuang_2000}.) We also use the controlled Pauli $Z$ gate (or $CZ = \ket{0} \bra{0} \otimes I + \ket{1} \bra{1} \otimes Z$) as two-qubit gate.

Let us explain about quantum circuits. An $n$-qubit quantum circuit is initialized to $\ket{0^n}$ and then arbitrary gates are applied to the state. We can apply gates at one time if each gate is applied to disjoint sets of qubits. The depth of a quantum circuit is $d$ if the whole operation of circuits can be decomposed to $U_d...U_2U_1$ where each $U_j$ is a tensor product of one- and two-qubit gates which act on disjoint sets of qubits.

Next, we describe measurements of quantum states. Mathematically, (projective) measurements are projections to some orthogonal bases. In this paper, We will use two kinds of measurements with the $X$ basis and the $Y$ basis. The orthogonal two states of $X$ basis are $\{\ket{+},\ket{-}\}$ and the ones of $Y$ basis are $\{\frac{\ket{0}+i\ket{1}}{\sqrt{2}},\frac{\ket{0}-i\ket{1}}{\sqrt{2}}\}$. Note that the measurements of $X$ and $Y$ basis are equivalent to the measurements of the computational basis if we apply $H$ and $HS$ gates before the measurements respectively.\vspace{2mm}

\noindent{\bf Quantum graph states.}
Quantum graph states are a certain type of entangled states corresponding to graphs first introduced by \cite{hein2004multiparty}. Let $G = (V,E)$ be a finite simple graph. Define an associated $|V|$ qubit graph state $\ket{\Phi_G}$ by
\begin{equation}\label{graphstate}
  \ket{\Phi_G} = \left(\prod_{e \in E} CZ_e\right)H^{\otimes|V|}\ket{0^{|V|}}.
\end{equation}
The graph state $\ket{\Phi_G}$ is a stabilizer state with stabilizer group generated by the operators $g_v = X_v\left(\prod_{w:\{w,v\}\in E} Z_w\right)$ for all $v \in V$.\vspace{2mm}

\noindent{\bf Classical circuits.}
A classical circuit is specified by a directed acyclic graphs. Vertices with no incoming and outgoing edges are inputs and outputs respectively and all other vertices are called gates. We must specify a function of each gate $\{0,1\}^k \rightarrow \{0,1\}$ where $k$ is the fan-in. We assume a classical probabilistic circuit receives an arbitrary binary $x$ as an input and outputs a binary $z$ and it also could input a random string $r$ drawn from some arbitrary distribution. We say an input bit and an output bit are correlated iff the value of the output bit depends on the value of the input bit. For each input bit $x_i$, we denote the lightcone $L_C(x_i)$ the set of output bits correlated with $x_i$ through a classical circuit $C$. Likewise, the lightcone $L_C(z_i)$ is the set of input bits correlated with an output bit $z_i$.  In this paper, we are interested in small-depth classical circuits with bounded fan-in. We also say that a classical probabilistic circuit solves the relation $R$ on all inputs if and only if the circuit takes any $x \in \{0,1\}^n$ and a random string $r$ as input and outputs $z \in \{0,1\}^m$ such that $xRz$ with high probability.
\section{Expander graphs and their properties}
Expander graphs are highly sparse but well connected graphs.  Notable applications of the graphs have been found in mathematics and computer science \cite{Barzdin1993,capalbo2002randomness,gromov2012generalizations}. Before giving the definition, we define the expansion ratio $h(G)$ of graph $G$. There are several ways to define the expansion ratio, for example edge expansion, vertex expansion and spectral expansion, but these are related to each other. The way we use is called vertex expansion. 

\begin{definition}[Expansion ratio]\label{doer}
    $h(G) = \mathrm{min}\left\{ \frac{|N_G(U)|}{|U|} \middle|\:\: U \subset V\mathrm{\ such \ that \ } 1 \leq |U| \leq \frac{1}{2}|V|\right\}$
\end{definition}
The definition of expander graphs we use in this paper is as follows.
\begin{definition}[Expander graphs, Definition 3.1.8 in \cite{kowalski2019introduction}]\label{def:exp}
    A family $(\Gamma_i)_{i \in \mathbb{N}}$ of finite non-empty connected graphs $\Gamma_i = (V_i,E_i)$ is an expander family, if there exist constants $d \geq 1$ and $h > 0$, independent of i, such that:
  \begin{enumerate}[(1)]
    \item The number of vertices $|V_i|$ ``tends to infinity'', in the sense that for any N $\geq 1$, there are only finitely many $i \in \mathbb{N}$ such that $\Gamma_i$ has at most $N$ vertices.

    \item For each $i \in \mathbb{N}$, we have $\max_{v \in V_i} deg(v) \leq d$, i.e., the maximum degree of the graphs is bounded independently of $i$.

    \item For each $i \in \mathbb{N}$, the expansion constant satisfies $h(\Gamma_i) \geq h > 0$, i.e., it is bounded away from 0 by a constant independent of i.
  \end{enumerate}
\end{definition}
$h$ and $d$ specify the family of expander graphs. The existence of expander graphs is by no means obvious, but it can be shown using probabilistic approaches or concrete constructions of such graphs \cite{kowalski2019introduction}. In this paper, an expander graph denotes a graph with sufficiently large $i$ of an expander family. 

We want to prove the robustness of expander graphs to arbitrary vertex removals in terms of the size of grid minor, which is Lemma \ref{minorinexpander}. We provide its proof in Appendix \ref{PL}.

\begin{lemma}\label{minorinexpander}
  Let $G=(V,E)$ be a expander graph. If we take a sufficiently small constant $\epsilon > 0$, the graph has a connected component $C$ which contains a $\Omega(|V|^{\frac{1}{4}})\times\Omega(|V|^{\frac{1}{4}})$ grid as a minor after up to an $\epsilon$ fraction of $V$ are adversarially removed.
\end{lemma}
\section{Quantum nonlocality of a triangle graph state}

In specific circumstances, local classical circuits cannot simulate measurement outcomes of entangled quantum states. This is called $quantum \ nonlocality$. In this section, we explain it occurs in a triangle graph state, as first shown in \cite{barrett2007modeling}. We will use this property to show the hardness of classical circuits to solve the relation problem in Section 5. 

\begin{figure}[htbp]
  \centering
  \includegraphics[clip,width=7.0cm]{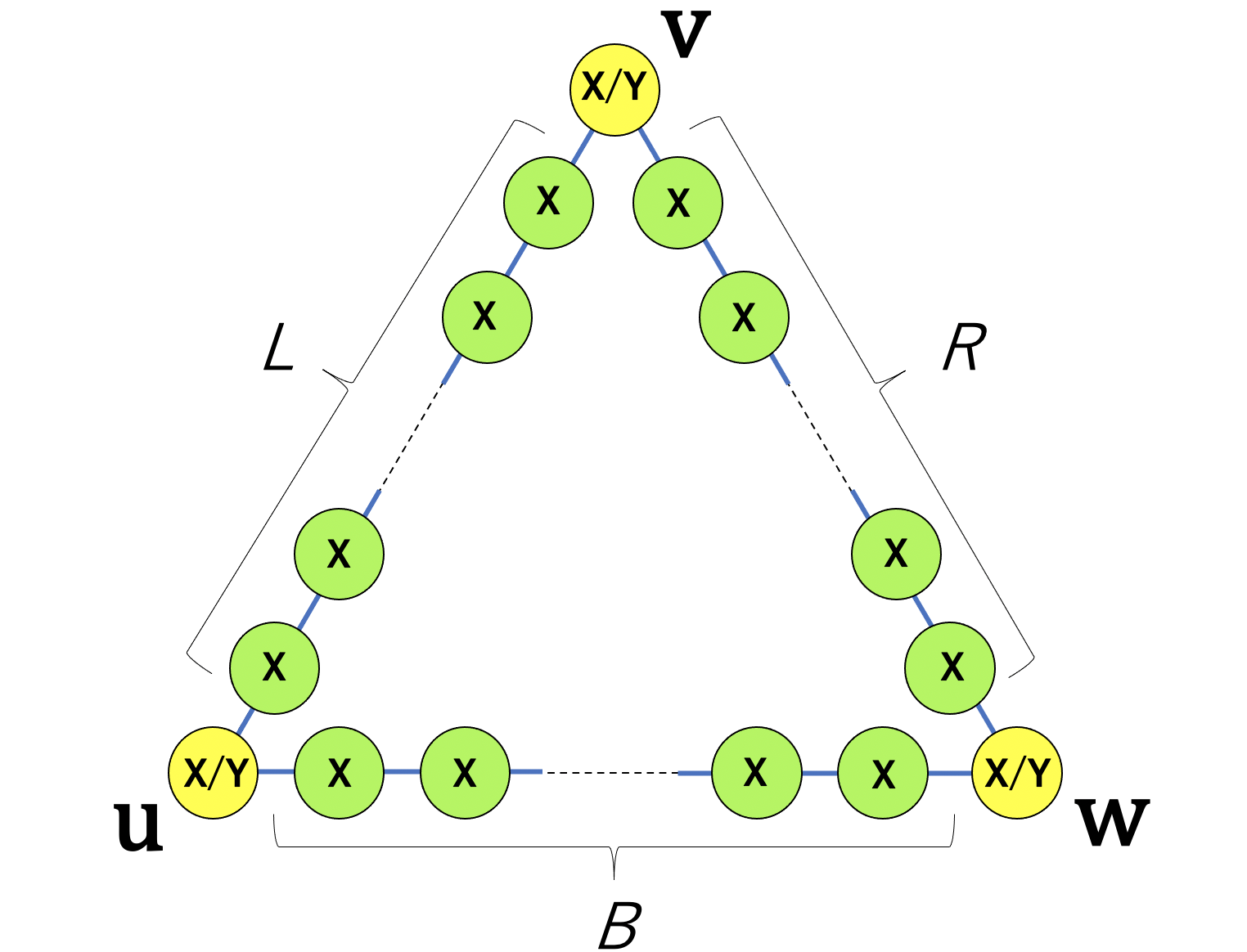}
  \caption{We consider an even cycle $\Gamma$ and three vertices $u,v,w$ such that all pairwise distances are even. $L,R,B$ are the region between the three vertices. Each qubit of the corresponding graph state is measured by the $X$ or $Y$ basis.}
  \label{triangle}
\end{figure}

First, we note about properties of a graph state in a triangle shape. Let $\Gamma$ be a triangle such that the distance between three vertices $u,v,w$ are all even and $\ket{\psi_\Gamma}$ be a graph state  for $\Gamma$ as in Equation~(\ref{graphstate}). We define $L,R,B$ as the vertices between $u$ and $v$, $v$ and $w$, $u$ and $w$, and $M$ as a total number of vertices in $\Gamma$. Also define $L_{odd},L_{even},R_{odd},R_{even},B_{odd},B_{even}$ as vertices of $L,R,B$ which have odd and even distance from $u,v,w$ respectively. The three bits $x = x_{u}x_{v}x_{w}$ decide the measurement basis of $u,v,w$ (the $X$ basis if $x_i$ = 0 and the $Y$ basis if $x_i$ = 1) and the other vertices are measured in the $X$ basis. We denote $\tau(x)$ possible measurement outcomes of all qubits of $\ket{\psi_\Gamma}$ for input $x$, i.e., $\tau(x) = \left\{ z \in \{0,1\}^{M} : \braket{z|H^{\otimes{M}}{S_u}^{x_u}{S_v}^{x_v}{S_w}^{x_w}|\psi_\Gamma} \neq 0 \right\}$. We consider a relationship between input $x$ and output $z \in \tau(x)$. Define the following summations:
\begin{equation}\nonumber
  z_L = \bigoplus_{i \in L_{odd}} z_i \hspace{10pt} z_R = \bigoplus_{i \in R_{odd}} \hspace{10pt} z_B = \bigoplus_{i \in B_{odd}} z_i \hspace{10pt} z_E = \bigoplus_{i \in \{u,v,w\} \cup R_{even} \cup L_{even} \cup B_{even}} z_i.
\end{equation}

\begin{claim}[Claim 3 in \cite{Bravyi_2018}]\label{nonlocalityintriangle}
  Let $x = x_u x_v x_w \in \{0,1\}^3$ and suppose $z \in \tau(x)$. Then $z_R \oplus z_B \oplus z_L=0$. Moreover, if $x_u \oplus x_v \oplus x_w = 0$ then
  \begin{eqnarray*}
    (x_u x_v x_w = 000) & z_E = 0, \hspace{33pt}
    (x_u x_v x_w = 110) & z_E \oplus z_L = 1, \\
    (x_u x_v x_w = 101) & z_E \oplus z_B = 1, \hspace{10pt} 
    (x_u x_v x_w = 011) & z_E \oplus z_R = 1.
  \end{eqnarray*}
\end{claim}

The following lemma is about quantum nonlocality in graph states of triangles and similar to Lemma~3 in Section 4.1 of \cite{Bravyi_2018}. It shows that when we assume a classical circuit has a kind of locality, the classical circuit cannot satisfy the relation of Claim \ref{nonlocalityintriangle} as inputs and outputs. We will give the proof in Appendix \ref{b}.

\begin{lemma}[\cite{barrett2007modeling,Bravyi_2018,le2019average}]\label{intersect}
  Consider a classical circuit which takes as an input a bit string $x = x_u x_v x_w \in \{0,1\}^3$ and a random string $r$, and outputs $z \in \{0,1\}^M$ which are corresponding to vertices of $\Gamma$. Let us assume output bits in $L$ depend on $r$ and at most one geometrically near input bit, which is either $x_u$ or $x_v$. Similarly, we assume output bits in $R$ and $B$ depend on $r$ and at most one geometrically near input bit.
  Then, the classical circuit $C$ cannot output $z \in \tau(x)$ with high probability.
\end{lemma}
\section{Proof of separation of depth between quantum and classical circuits to solve Graph State Sampling problem}\label{sec:proof}

In this section, we define Graph State Sampling problem
and prove a separation of depth between quantum circuits and classical circuits. We consider an almost all induced subgraph $S$ of a graph $G \times K_2$ such that $G$ is an expander graph. Then, we prove $\rho(G \times K_2)$ can be solved on all inputs by a constant-depth quantum circuit on $\Theta(|G \times K_2|)$ qubits, but $\Omega(\log|S|)$ depth is required for any classical probabilistic circuits to solve $\rho(S)$ on all inputs.

\subsection{Definition of Graph State Sampling problem $\rho(G)$}\label{def:GSS}
In this subsection, for any graph $G$, we define Graph State Sampling problem $\rho(G)$. 

The relation is defined as a subset of $\{0,1\}^{|V|+|E|} \times \{0,1\}^{|V|}$ and thus consists of pairs $(x,z)$, where $x \in \{0,1\}^{|V|+|E|}$ represents the input and $z\in\{0,1\}^{|V|}$ represents the output.
Each bit of $x$ corresponds to a vertex or a edge. The string $x$ decides the quantum graph state and the measurement bases: a $CZ$ gate corresponding to edge $e$ is applied if $x_e = 1$, and the qubit corresponding to a vertex $v$ is measured in the~$X$ basis if $x_v$ is 0, or in the $Y$ basis if $x_v$ is 1. The quantum state $\ket{\psi_x}$ for each $x$ before the measurement in the computational basis is thus:
\begin{equation*}
  \ket{\psi_x} = H^{\otimes|V|} \prod_{x_v=1} S_v \left(\prod_{x_e = 1} CZ_{e}\right)H^{\otimes|V|}\ket{0^{|V|}}.
\end{equation*}
The output $z \in \{0,1\}^{|V|}$ of the relation is any possible outcome of the measurement of this quantum state (note that there are possibly several measurement outcomes $z$ for each $x$). Since the probability of measurement results of each binary string $z$ is $|\braket{z|\psi_x}|^2$, the definition of $\rho(G)$ is as follows. 
\begin{definition}
  Given a graph G,
  \begin{equation}\nonumber
    \rho(G) = \{(x,z)| x \in \{0,1\}^{|V|+|E|} \ \mathrm{and} \ z \in \{0,1\}^{|V|} \mathrm{ \ such \ that \ } |\braket{z|\psi_x}|^2>0 \}.
  \end{equation}
\end{definition}

\subsection{Constant-depth quantum circuits to solve Graph State Sampling problem}
The following is easily shown by Lemma \ref{edgecoloring}.

\begin{lemma}\label{constantquantum}
  When the maximum degree of graph G = (V,E) is bounded by a constant, $\rho(G)$ on all inputs can be solved with certainty by a constant-depth quantum circuit on $\Theta(|G|)$ qubits composed of one- and two-qubit gates.
\end{lemma}
\begin{proof}
  The initial state $\ket{x} \otimes \ket{0^{|V|}}$ is prepared on $|E|+2|V| = \Theta(|G|)$ qubits.
  We apply $H^{\otimes|V|}$ to the last $|V|$ qubits and then controlled-$CZ$ ($CCZ$) gates and controlled-$S$ ($CS$) gates that apply $CZ_e$ if $x_e = 1$ and $S_v$ if $x_v = 1$. Since $G$ is edge colorable with a constant number from Lemma \ref{edgecoloring} and $CCZ$ gates corresponding to edges assigned the same color can be applied simultaneously (since they act on disjoint sets of qubits), the total depth of $CCZ$ gates can be bounded by a constant. $CS$ gates can also be applied in constant depth. We finally apply $H^{\otimes|V|}$ to the last $|V|$ qubits and measure these qubits in the computational basis, which gives a string $z$ such that $(x,z)\in \rho(G)$. 
  
  The total depth of this circuit can be bounded by a constant. (Note that each $CCZ$ and $CS$ gate can be implemented in constant depth using our elementary gates \cite{amy2013meet}.) We refer to Figure \ref{cdqc} for an illustration.
\end{proof}

\begin{figure}[htbp]
  \centering
  \includegraphics[clip,width=10.0cm]{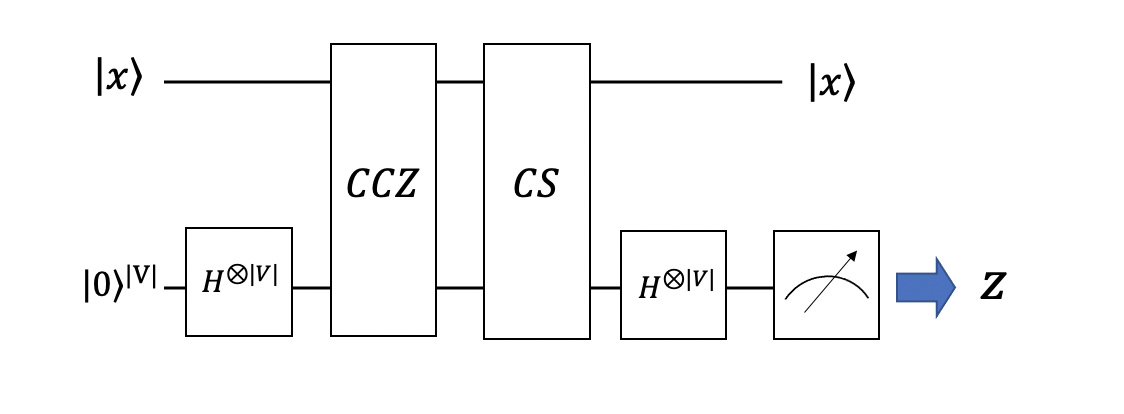}
  \caption{Constant-depth quantum circuit to solve $\rho(G)$.}
  \label{cdqc}
\end{figure}

\subsection{Hardness to solve Graph State Sampling problem with shallow classical circuits}

In this subsection, we prove Theorem \ref{main}, and especially the classical hardness. 
The impossibility argument in Section 4 assumed classical circuits had a kind of geometrical locality. The lower bounds in this section, however, do not require any geometrical locality of shallow classical circuits.
The proofs are similar to the proofs of the results of Section 4.2 in \cite{Bravyi_2018}, with the notable exception of Claim \ref{exception} and the discussion afterwards (in particular, the definition of boxes), which are specifically tailored for the expander graphs we consider.

\subsubsection{Good and bad vertices}

To begin with, we define ``good'' and ``bad'' vertices. In a shallow classical circuit, most input bits are not correlated with many output bits. We call a vertex ``bad'' if the corresponding input bit are correlated with many output bits. Here is the formal definition.
\begin{definition}
  Given a graph $G = (V,E)$, we consider $\rho(G)$ and a classical probabilistic circuit $C$ for it. Then, a vertex $v \in V$ is good if $L_C(x_v) = O(|V|^{\frac{1}{16}})$ and bad if $v$ is not good.
\end{definition}
The following claim is similar to Claim 5 in \cite{Bravyi_2018}.
\begin{claim}\label{bad}
  Let $G = (V,E)$ be a graph such that the maximum degree is bounded by a constant and $C$ be a classical probabilistic circuit for $\rho(G)$. Suppose the fan-in is bounded by a constant $K$ and the depth $d$ is less than $\frac{\log|V|}{32 \log{K}}$ , the number of bad vertices is $o(|V|)$. 
\end{claim}

\begin{proof}
  Since the number of correlated input bits increases by at most $K$ times when the depth increases by~1,
  \begin{equation}\label{zlimit}
    |L_C(z_i)| \leq K^d < |V|^{\frac{1}{32}} \hspace{10pt} \mathrm{for} \ \mathrm{all} \ v \in V.
  \end{equation}
  Let us consider a bipartite graph whose vertices are respective bits of $x$ and $z$, and a edge is spanned if and only if $x_i$ and $z_j$ are correlated. Since the maximum degree of $G$ is bounded by a constant, we have $|x| = |V| + |E| = \Theta(|V|)$. From Equation (\ref{zlimit}) and considering edges spanned from $z$, the total number of edges is limited by $|V| \cdot |V|^{\frac{1}{32}}$. Since bad vertices are correlated with $\Omega(|V|^{\frac{1}{16}})$ output bits, the total number of $x_v$ such that $v$ is bad is $O(|V|^{\frac{31}{32}})$ and this means the number of bad vertices is $o(|V|)$.
\end{proof}

\subsubsection{Proof of Theorem \ref{main}}

Before the proof, we rewrite Theorem \ref{main} using the notations we defined. The reason we consider the graph product $G \times K_2$ is to take a cycle which has even length for using Lemma~\ref{intersect}.
\begin{theorem}\label{reTh}
 There exist constants $\alpha > 0$ and $\epsilon > 0$ such that the following holds for all sufficiently large expander graphs $G$:
 \begin{enumerate}[(i)]
     \item $\rho(G \times K_2)$ can be solved on all inputs with certainty by a constant-depth quantum circuit on $\Theta(|G \times K_2|)$ qubits composed of one- and two-qubit gates;
     \item for any induced subgraph $S$ such that $|S| \geq (1-\epsilon)|G \times K_2|$, no bounded-fanin classical probabilistic circuits whose depth is less than $\alpha \log(|S|)$ can solve $\rho(S)$ on all inputs.
 \end{enumerate}
\end{theorem}

\begin{proof}[Proof of Theorem \ref{reTh}]
 We can take $|G| = |V|$ arbitrary large from the property of expander graphs. Then $|G \times K_2| = 2|G|$ and $|S| \geq (1 - \epsilon)|G \times K_2|$ are also sufficiently large. Let us introduce some convenient notations. Given a vertex $u\in G$, we denote $u'$ and $u''$ the two corresponding vertices in $G\times K_2$ (with no special order). Given a vertex $v\in G\times K_2$, we denote $\bar{v}$ the other vertex in $G\times K_2$ associated to the same vertex in $G$.
 
 First, we prove Theorem \ref{reTh} (i). We consider the Graph State Sampling problem $\rho(G \times K_2)$. Since the maximum degree of $G$ is bounded by a constant $d$, the maximum degree of $G \times K_2$ is bounded by $d + 1$. From Lemma~\ref{constantquantum}, $\rho(G \times K_2)$ can be solved with a constant-depth quantum circuit.

 Next, we will prove Theorem \ref{reTh} (ii), the hardness to solve $\rho(S)$ on all inputs with shallow classical circuits. Let $C$ be a classical probabilistic circuit to solve $\rho(S)$ on all inputs and $K$ be the bounded fan-in of $C$. In order to reach a contradiction, we assume the depth of $C$ is less than $\frac{\log|S|}{32 \log{K}}$. Then the number of bad vertices is small, which enables us to prove the following claim. 

 \begin{claim}\label{exception}
  $S$ contains an induced subgraph $G' \times K_2$ such that all vertices are good and $G'$ contains a $\Omega(|V|^{\frac{1}{4}})\times\Omega(|V|^{\frac{1}{4}})$ grid as a minor.
 \end{claim}

 \begin{proof}
  From Claim \ref{bad}, we know that $S$ contains $o(|S|)$ bad vertices. Let us remove all these bad vertices, and write $S_{good}$ the remaining set. We further remove all vertices $u\in S_{good}$ such that $\bar{u}\notin S_{good}$. The remaining set of vertices induces a graph $H\times K_2$, for an induced subgraph $H$ of $G$ such that $|H|\ge (1-2\epsilon-o(1))|G|$. From Lemma \ref{minorinexpander}, when $\epsilon$ is taken small enough, the graph $H\times K_2$ has a connected component $G' \times K_2$, where $G'$ is contains a $\Omega(|V|^{\frac{1}{4}})\times\Omega(|V|^{\frac{1}{4}})$ grid as a minor.
 \end{proof}

 Remember the definition of a graph minor (Definition 2 in Section 2.1). Each connected subgraph $G_u$ in $G'$ forming the grid (except connected subgraphs on the corners of the grid) is adjacent to the four connected subgraphs $G_{v_1}, G_{v_2}, G_{v_3}, G_{v_4}$ where $\{u,v_1\}$, $\{u,v_2\}$, $\{u,v_3\}$ and $\{u,v_4\}$ are edges of the grid. For each $G_u$, we arbitrarily select two of these four components. Assume for instance that we selected $G_{v_1}$ and $G_{v_2}$. We choose arbitrarily one vertex $v$ in $G_u$ adjacent to a vertex of $G_{v_1}$, and one vertex $w$ in $G_u$ adjacent to a vertex of $G_{v_2}$. We then arbitrarily choose one path inside $G_u$ that connects $v$ and $w$ (such a path necessarily exists). We refer to Figure \ref{path} for an illustration. 
 
 \begin{figure}[htbp]
  \centering
  \includegraphics[clip,width=9.0cm]{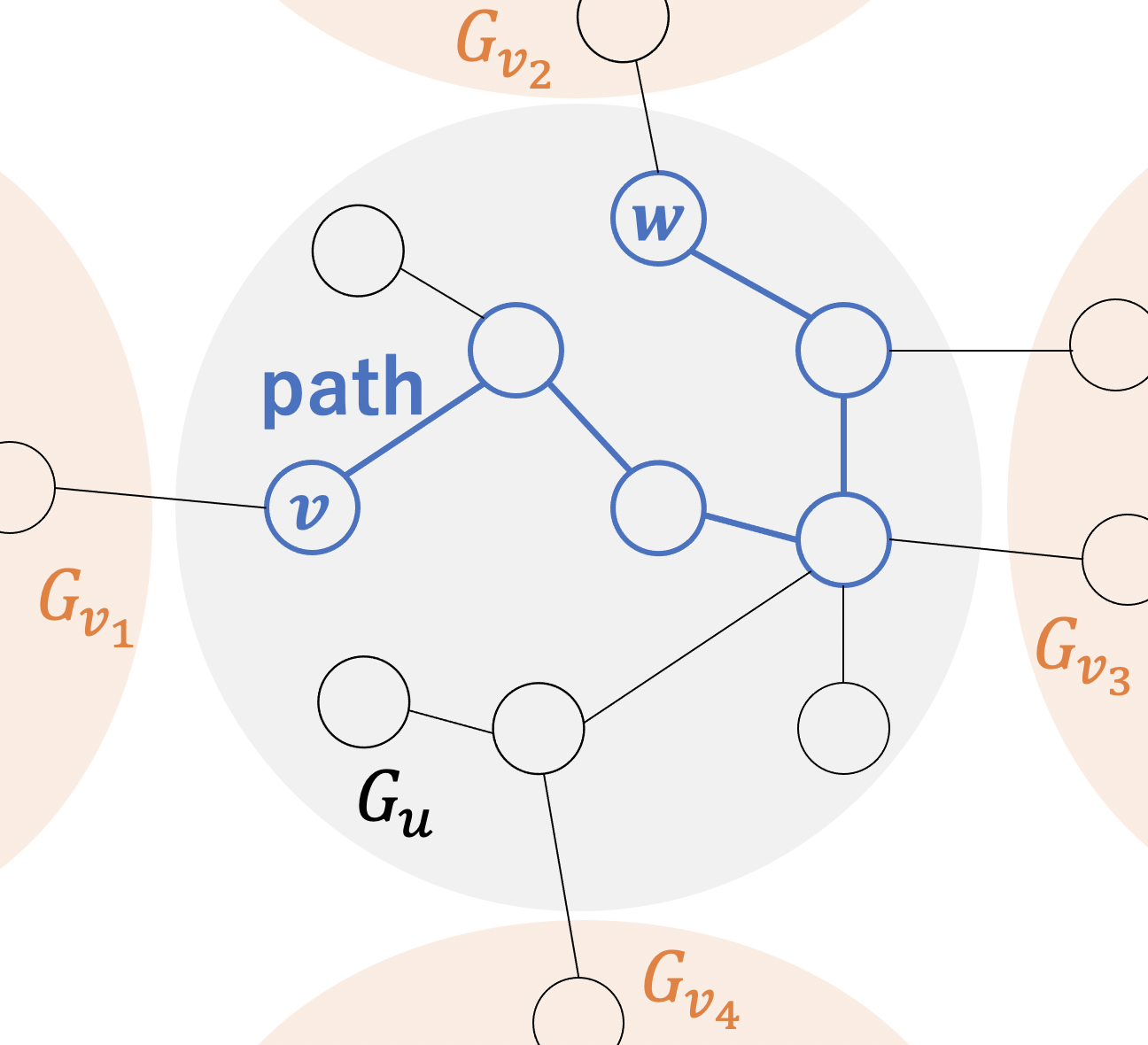}
  \caption{The path in $G_u$ described with blue line connects $G_{v_1}$ and $G_{v_2}$.}
  \label{path}
 \end{figure}

 We denote T the length of one side of the grid (T = $\Omega(|V|^{\frac{1}{4}})$). From each connected subgraph forming the $T \times T$ grid of $G'$, we choose one vertex such that it is on a path selected in the way above (this condition of the vertices is required when adding missing segments inside boxes for Claim \ref{cycle}). For such a vertex $j$, we define $\mathrm{Box}(j) \subseteq G' \times K_2$ as a 2D grid of connected subgraphs of size $\lfloor |V|^{\frac{1}{8}} \rfloor \times \lfloor |V|^{\frac{1}{8}} \rfloor$ centered at the connected subgraph which $j'$ and $j''$ belong to. We choose grid-shaped regions $P,Q,R \subseteq G' \times K_2$ as shown in Figure~\ref{grid}.
 $P$ is a upper-left region of the graph product of $K_2$ and a $\lfloor T/3 \rfloor \times \lfloor T/3 \rfloor$ grid of connected subgraphs cut from $G' \times K_2$. $Q$ is a upper-right region and $R$ is a bottom-left region. The following claim is similar to Claim 6 in \cite{Bravyi_2018}.
 
\begin{figure}[htbp]
 \begin{minipage}{0.49\hsize}
  \begin{center}
   \includegraphics[width=75mm]{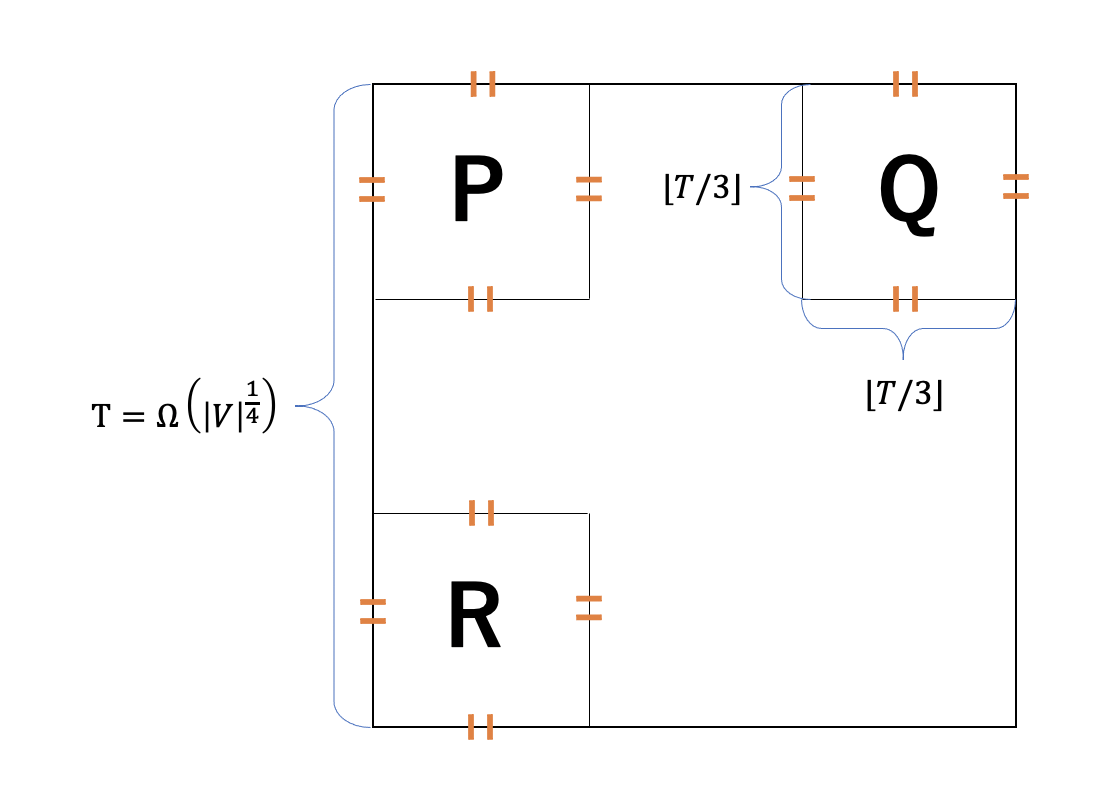}
  \end{center}\vspace{-5mm}
  \caption{Definition of the regions $P,Q,R$ of $G' \times K_2$.}
  \label{grid}
 \end{minipage}
 \hspace{2mm}
 \begin{minipage}{0.49\hsize}
  \begin{center}
   \includegraphics[width=55mm]{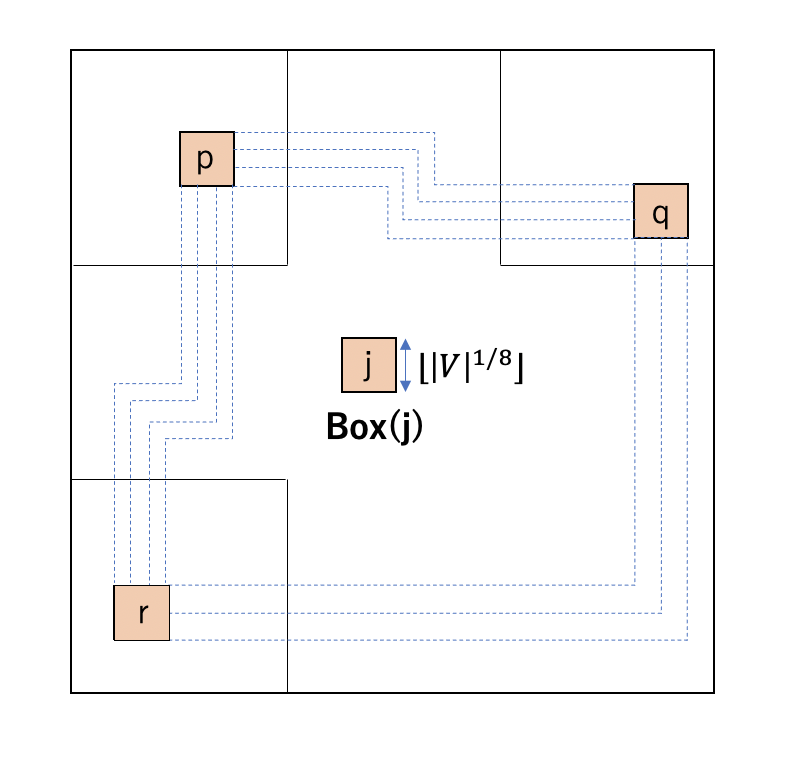}
  \end{center}\vspace{-3mm}
  \caption{Definition of Box($j$) and a possible choice of $p,q,r$.}
  \label{box}
 \end{minipage}
\end{figure}

\newpage
\begin{claim}
  For all large enough $|V|$, we can choose a triple of vertices $p,q,r \in V$ such that $p', p'' \in P$, $q', q'' \in Q$, $r', r'' \in R$ and
  \begin{eqnarray}
    \mathrm{Box}(p) \subseteq P, \ \mathrm{Box}(q) \subseteq Q, \ \mathrm{Box}(r) \subseteq R, \hspace{40pt} \label{boxes} \\
    L_C(x_{p'} \cup x_{p''}) \cap \mathrm{Box}(q) = \emptyset, \ L_C(x_{p'}\cup x_{p''}) \cap \mathrm{Box}(r) = \emptyset, \label{u} \\
    L_C(x_{q'} \cup x_{q''}) \cap \mathrm{Box}(p) = \emptyset, \ L_C(x_{q'} \cup x_{q''}) \cap \mathrm{Box}(r) = \emptyset, \label{v} \\
    L_C(x_{r'} \cup x_{r''}) \cap \mathrm{Box}(p) = \emptyset, \ L_C(x_{r'} \cup x_{r''}) \cap \mathrm{Box}(q) = \emptyset. \label{w}
  \end{eqnarray}
\end{claim}

 \begin{proof}
  One connected subgraph in the grid minor can belong to at most $|V|^{\frac{1}{8}} \times |V|^{\frac{1}{8}} = |V|^{\frac{1}{4}}$ boxes. Since the vertices in $G' \times K_2$ are good, a given lightcone $L_C(x_{u'} \cup x_{u''})$ can intersect with at most $|V|^{\frac{1}{4}} \times |L_C(x_{u'} \cup x_{u''})| = |V|^{\frac{1}{4}} \times O(|S|^{\frac{1}{16}}) = O(|V|^{\frac{5}{16}})$ boxes. The number of possibilites to choose a box in $Q$ is $\Omega(|V|^{\frac{1}{4}}) \times \Omega(|V|^{\frac{1}{4}}) = \Omega(|V|^{\frac{1}{2}})$. Thus if we pick boxes uniformly at random then
  \begin{equation}
    \mathrm{Pr}[L_C(x_{p'} \cup x_{p''}) \cap \mathrm{Box}(q) = \emptyset] \leq O \left(\frac{|V|^{\frac{5}{16}}}{|V|^{\frac{1}{2}}}\right) < \frac{1}{6}
  \end{equation}
  for large enough $|V|$. A similar bound applies to the five others that appear in the three equations (\ref{u}, \ref{v}, \ref{w}). By the union bound, there exists at least one choice of $p,q,r$ that satisfies all the four equations (\ref{boxes}, \ref{u}, \ref{v}, \ref{w}).
 \end{proof}

  Below we consider a cycle $\Gamma$ that is a subgraph of $G' \times K_2$. The following claim is similar to Claim 7 in~\cite{Bravyi_2018}.

 \begin{claim}\label{cycle}
  The following holds for all sufficiently large $|V|$.  Fix some triple of vertices $p,q,r$ satisfying the four equations (\ref{boxes}, \ref{u}, \ref{v}, \ref{w}). Then there exists an even length cycle $\Gamma$ containing $p',p'',q',q'',r',r''$ such that the lightcones $L_C(x_{p'} \cup x_{p''})$, $L_C(x_{q'} \cup x_{q''})$, $L_C(x_{r'} \cup x_{r''})$ contain no vertices of $\Gamma$ lying outside of $\mathrm{Box}(p) \cup \mathrm{Box}(q) \cup \mathrm{Box}(r)$.  
 \end{claim}

 \begin{proof}
  Since the size of connected subgraphs of each box is $\lfloor |V|^\frac{1}{8} \rfloor \times \lfloor |V|^\frac{1}{8} \rfloor$, we can choose $\lfloor |V|^\frac{1}{8} \rfloor$ pairwise vertex disjoint paths $\gamma$ that connect any pair of boxes Box($p$), Box($q$), Box($r$) (in each connected subgraph, we can always find a path which connects adjacent connected subgraphs). We refer to Figure \ref{box} for an illustration. Let $\gamma(a,b)$ be a path connecting $\mathrm{Box}(a)$ and $\mathrm{Box}(b)$, where $a \neq b \in \{p,q,r\}$. Any triple of paths $\gamma(p,q),\gamma(q,r),\gamma(p,r)$ can be completed to a cycle $\Gamma$ by adding the missing segments of the cycle inside the boxes Box($p$), Box($q$), Box($r$) because $p,q,r$ are defined to take such a path inside each box. Since all vertices are good in $G' \times K_2$ and each connected subgraph belongs to at most one path $\gamma$, we infer that $L_C(x_{p'} \cup x_{p''})$ intersects with at most $2 \cdot O(|S|^{\frac{1}{16}}) = O(|V|^{\frac{1}{16}})$ paths $\gamma$. Thus if we pick the path $\gamma(p,q)$ uniformly at random among all $\lfloor |V|^\frac{1}{8} \rfloor$ possible choices then
  \begin{equation}
    \mathrm{Pr}[L_C(x_{p'} \cup x_{p''}) \cap \gamma(p,q) \neq \emptyset] \leq O\left(\frac{|V|^{\frac{1}{16}}}{|V|^{\frac{1}{8}}}\right) < \frac{1}{9}
  \end{equation}
  for enough large $|V|$. The same bound applies to eight remaining combinations of lightcones $L_C(x_{p'} \cup x_{p''}), L_C(x_{q'} \cup x_{q''}), L_C(x_{r'} \cup x_{r''})$ and paths $\gamma(p,q)$, $\gamma(p,r)$, and $\gamma(q,r)$. By the union bound, there exists at least one triple of paths $\gamma(p,q), \gamma(q,r) ,\gamma(p,r)$ that do not intersect with $L_C(x_{p'} \cup x_{p''}), L_C(x_{q'} \cup x_{q''}), L_C(x_{r'} \cup x_{r''})$. When we choose $v'$ and $v''$ consecutively, any cycle in $G' \times K_2$ has an even length. Since we can take the cycle in the way above, the cycle has a even length and is the desired cycle $\Gamma$.
 \end{proof}
 Let $p,q,r$ and $\Gamma$ be chosen as described in Claim \ref{cycle}. Let $M$ be the even number of vertices of $\Gamma$. When we choose $p',p'',q',q'',r',r''$ properly, the distances between $p',q',r'$ are all even. Consider the subset of instances where
 \begin{eqnarray*}
  x_e = 
  \left\{
  \begin{array}{l}
  1 \hspace{10pt} \mathrm{if \ } e \mathrm{ \ is \ an \ edge \ of \ } \Gamma\\
  0 \hspace{10pt} \mathrm{otherwise}
  \end{array}
  \right.
  \mathrm{and}
  \hspace{10pt}
  x_v = 0 \ \mathrm{if} \ (v \in V \setminus \{p',q',r'\})
 \end{eqnarray*}
 There are $2^3 = 8$ such instances corresponding to choices of input bits $x_{p'},x_{q'},x_{r'} \in \{0,1\}$. Let us fix inputs $x$ of the circuit $C$ except $\{x_{p'},x_{q'},x_{r'}\}$ and consider only output bits $z_j$ with $j \in \Gamma$. By the way of fixing, we obtain a classical circuit $D$ which takes a three-bit string $x_{p'} x_{q'} x_{r'} \in \{0,1\}^3$ and a random string $r$ as input and output $z_\Gamma \in \{0,1\}^M$. For any input bit $x_i \in \{x_{p'}, x_{q'}, x_{r'}\}$ we have $L_D(x_i) \subseteq L_C(x_i)$ since any pair of input and output variables which are correlated in $D$ are also correlated in $C$, by definition. Our assumption that $C$ can solve $\rho(S)$ on all inputs implies that $D$ can output $z_\Gamma \in \tau({x_{p'}}{x_{q'}}{x_{r'}})$. From Lemma \ref{intersect}, at least one output bit $z_j$ such that $j \in \Gamma$ and $j \notin \{p', q', r'\}$ depends on the two geometrically near inputs from $x_{p'}, x_{q'}, x_{r'}$. By $L_D(x_i) \subseteq L_C(x_i)$, the same is true for the input-output dependency of $C$. From Claim 9, for each $x_i$, $L_C(x_i)$ only intersects with $z_j$ such that $j$ $\in \mathrm{Box}(i)$ and there is a contradiction. Therefore, the depth of any classical probabilistic circuit that has bounded fan-in $K$ and solves $\rho(S)$ on all inputs is not less than $\frac{\log |S|}{32 \log K}$. This concludes the proof of Theorem 3.
\end{proof}

\bibliography{lipics-v2021-sample-article}

\appendix
\section{Proof of Lemma \ref{minorinexpander}}\label{PL}

First, we introduce a notation. The range in Definition \ref{doer}, 1 $\leq |U| \leq \frac{1}{2}|V|$, may be a little arbitrary. In terms of the range where the expansion ratio is considered, we define more general expander graphs as follows.
\begin{definition}[Definition 2.2 in \cite{krivelevich2019expanders}]\label{doie}
  Let $G = (V,E)$ be a graph, let $I$ be a set of positive integers. The graph G is an $I$-$expander$ if a positive constant $h$ exists such that $N_G(U) \geq h|U|$ for every vertex subset $U \subset V$ satisfying $|U| \in I$.
\end{definition} 
\noindent Note that this definition does not limit the maximum degree of graphs. When the graph $G$ is an expander graph (as defined as Definition \ref{def:exp}), $G$ is also a $\left[1,\frac{|V|}{2}\right]$-expander with bounded degree.

Then, the following claim shows an expander graph still has a large connected component and it can be described using the notation of Definition \ref{doie} when a small fraction of vertices are adversarially removed. 
\begin{claim}
  Let $G=(V,E)$ be an expander graph. If we take a sufficiently small constant $\epsilon > 0$, the graph has a connected component $C$ which has more than $\frac{|V|}{2}$ vertices and is a $\left[\frac{|C|}{3},\frac{2|C|}{3}\right]$-expander after up to an $\epsilon$ fraction of the vertices are adversarially removed.
\end{claim}

\begin{proof}
Let $h$ be the expansion ratio and $d$ be the maximum degree of the graph $G$.
  Let $V' \subset V$ be an arbitrary subset of vertices such that $|V'| \leq \epsilon|V|$. Let $C_1,...,C_m$ be the connected components of the left graph after $|V'|$ vertices are adversarially removed. 
  
  To reach a contradiction, we assume every connected component is equal to or smaller than half of $|V|$, i.e., for all $i, \ |C_i| \leq \frac{|V|}{2}$. From $|V \setminus V'| + |V'| = |V|$, $|V \setminus V'| = |V| - |V'| \geq (1-\epsilon)|V|$. Each connected component $C_i$ has its neighbor $N_G(C_i)$ such that $|N_G(C_i)| \geq h|C_i|$ since $|C_i| \leq \frac{|V|}{2}$. A vertex can be a neighbor of at most $d$ connected components at the same time. Therefore, by the summation of neighbors of all connected components $C_i$, 
  \[
   |N_G(V \setminus V')| \geq \frac{h|V \setminus V'|}{d} \geq \frac{h}{d} (1-\epsilon) |V|. 
  \]
  In terms of the total number of vertices, $|V \setminus V'| + |N_G(V \setminus V')| \leq |V|$. Then,
  \[
    \epsilon|V| \geq |V'| = |V| - |V \setminus V'| \geq |N_G(V \setminus V')| \geq \frac{h}{d}(1-\epsilon)|V|.
  \]
  Thus, $\epsilon \geq \frac{\frac{h}{d}}{1+\frac{h}{d}}$ and this contradicts $\epsilon$ is sufficiently small.  Therefore we can pick a connected component $C$ such that $|C| > \frac{|V|}{2}$ from $C_1,...,C_m$.
  
  In the graph $C$, we consider an arbitrary subset of vertices $W$, which satisfies $\frac{|C|}{3} \leq W \leq \frac{2|C|}{3}$. From the expander property of $G$, $|N_G(W)| \geq \frac{h|C|}{3} > \frac{h|V|}{6}$. Since the number of removed vertices is up to $\epsilon|V|$, $|N_C(W)| > \left(\frac{h}{6}-\epsilon\right)|V|$. If we take $\epsilon$ smaller than $\frac{h}{6}$, $C$ is a $\left[\frac{|C|}{3},\frac{2|C|}{3}\right]$-expander.
\end{proof}

The next claim is to show there is a relation between $I$-expanders of Definition 5.

\begin{claim}[Lemma 2.4 in \cite{krivelevich2019expanders}]
  Let graph $G = (V,E)$ be a $\left[\frac{|V|}{3},\frac{2|V|}{3}\right]$-expander. Then there is a vertex subset $Z \subset V$ such that $|Z| < \frac{|V|}{3}$ and the graph $G' = G[V \setminus Z]$ is a $\left[1,\frac{|G'|}{2}\right]$-expander. 
\end{claim}

The most significant result of minors of expander graphs is Claim \ref{claim} below. It is known that this bound ($O(\frac{|V|}{\log(|V|)})$) is tight especially in terms of the size of grid minors.

\begin{claim}[Corollary 8.3 in \cite{krivelevich2019expanders} and Theorem 1.1 in \cite{chuzhoy2019large}]\label{claim}
  Let graph $G = (V,E)$ be a $\left[1,\frac{|V|}{2}\right]$-expander. For any graph $H$ with $O(\frac{|V|}{\log(|V|)})$ vertices and edges, G contains H as a minor.
\end{claim}

Finally, using the above claims, we can prove Lemma \ref{minorinexpander}.

\begin{proof}[Proof of Lemma \ref{minorinexpander}]
  From Claim 20, after the removal, the left graph has a connected component $C$ which is a $\left[\frac{|C|}{3},\frac{2|C|}{3}\right]$-expander. Using Claim 21, $C$ contains an induced subgraph $C'$ such that $C'$ is a $\left[1,\frac{|C'|}{2}\right]$-expander and $|C'| > \frac{2|C|}{3} > \frac{|V|}{3}$. When $n$ is sufficiently large, $\frac{n}{\log(n)} \gg n^{\frac{1}{2}}$. Therefore, from Claim 22, $C$ contains a $\Omega(|V|^{\frac{1}{4}})\times \Omega(|V|^{\frac{1}{4}})$ grid as a minor since the maximum degree of grid graphs is 4, which is a constant, and the number of vertices and edges of a $\Omega(|V|^{\frac{1}{4}})\times \Omega(|V|^{\frac{1}{4}})$ grid is $\Omega(|V|^\frac{1}{2})$.
\end{proof}

\section{Proof of Lemma \ref{intersect}}\label{b}
We note a mathematical claim below.
\begin{claim}[\cite{barrett2007modeling,Bravyi_2018,le2019average}]\label{affine} 
  Consider any affine function $q:\{0,1\}^3 \rightarrow \{0,1\}$ and any three affine function $q_1:\{0,1\}^2\rightarrow\{0,1\},q_2:\{0,1\}^2\rightarrow \{0,1\},q_3:\{0,1\}^2\rightarrow\{0,1\}$ such that
  \begin{equation}\nonumber
    q_1(b_2,b_3) \oplus q_2(b_1,b_3) \oplus q_3(b_1,b_2) = 0 \label{1}
  \end{equation}
  holds for any $(b_1,b_2,b_3) \in \{0,1\}^{3}$. Then at least one of the four following equalities does not hold:
  \begin{eqnarray*}\nonumber
    &\mspace{-75mu}q(0,0,0)=0, \\
    &q(0,1,1) \oplus q_1(1,1) = 1, \\
    &q(1,0,1) \oplus q_2(1,1) = 1, \\
    &q(1,1,0) \oplus q_3(1,1) = 1.
  \end{eqnarray*}
\end{claim}
\begin{proof}[Proof of Lemma \ref{intersect}]
Let us write $z = F(x,r)$ for the function which is computed by the circuit. Below we show for each $r$ there exists a $x$ such that $F(x,r) \notin \tau(x)$. Let $r$ be fixed and we consider $z = F(x,r)$ as a function of $x$. Suppose first that $z_R \oplus z_B \oplus z_L=1$ for some $x_0 \in \{0,1\}^3$. Then by Claim \ref{nonlocalityintriangle}, $F(x_0,r) \notin \tau(x_0)$ and we are done. Next suppose that $z_R \oplus z_B \oplus z_L=0$ for all $x \in \{0,1\}^3$. From the assumption, $z_E$ is an affine function of $x_u$, $x_v$ and $x_w$. In the same way, $z_L$ is an affine function $x_u$ and $x_v$. $z_R$ is an affine function of $x_v$, $x_w$ and $z_B$ is one of $x_u$, $x_w$. Therefore, we can correspond $z_E$, $z_L$, $z_R$ and $z_B$ to $q, q_1, q_2$ and $q_3$ in Claim \ref{affine} and, from Claim \ref{affine}, the output $z$ cannot meet all the four equation (2)-(5) simultaneously. Therefore, the classical circuit cannot output $z \in \tau(x)$ with high probability.
\end{proof}
\end{document}